\documentclass[10pt,letterpaper]{article}
\usepackage{opex3}

\usepackage{graphicx}
\usepackage{dcolumn}
\usepackage{amsfonts,amsthm,amssymb,amsmath}
\usepackage{epsfig,float,cite}
\usepackage{bm}

\usepackage[dvips,colorlinks=true,linkcolor=black,citecolor=black,urlcolor=black,bookmarks=false]{hyperref}

\newcommand{\ket}[1]{\ensuremath{|\,{#1}\,\rangle}}

\newcommand{\brm}[1]{\ensuremath{\boldsymbol{#1}}}

\newcommand{\esp}[1]{\ensuremath{\langle\,{#1}\,\rangle}}

\newcommand{\sinc}{\ensuremath{\mbox{\hspace{1.3pt}sinc}\,}}
\newcommand{\sint}{\ensuremath{\mbox{\hspace{1.3pt}sint}\,}}

\newcommand{\Si}{\ensuremath{\mbox{\hspace{1.3pt}Si}\,}}

\begin{document}

\title{Quantifying the non-Gaussianity of the state of spatially correlated down-converted photons}

\author{E. S. G\'{o}mez,$^{1,2,3}$ W. A. T. Nogueira,$^{1,2,3}$ C. H. Monken,$^{4}$ and G.~Lima$^{1,2,3,*}$}
\address{$^1$Departamento de F\'{i}sica, Universidad de Concepci\'{o}n, 160-C Concepci\'{o}n, Chile\\
$^2$Center for Optics and Photonics, Universidad de Concepci\'{o}n, Chile\\
$^3$MSI-Nucleus on Advanced Optics, Universidad de Concepci\'{o}n, Chile\\
$^4$Departamento de F\'{i}sica, ICEx, Universidade Federal de Minas Gerais, Belo Horizonte, Brazil\\
\href{mailto:glima@udec.cl}{\textcolor{blue}{$^*$glima@udec.cl}}}

\begin{abstract}
The state of spatially correlated down-converted photons is usually treated as a two-mode Gaussian entangled state. While intuitively this seems to be reasonable, it is known that new structures in the spatial distributions of these photons can be observed when the phase-matching conditions are properly taken into account. Here, we study how the variances of the near- and far-field conditional probabilities are affected by the phase-matching functions, and we analyze the role of the EPR-criterion regarding the non-Gaussianity and entanglement detection of the spatial two-photon state of spontaneous parametric down-conversion (SPDC). Then we introduce a statistical measure, based on the \emph{negentropy} of the joint distributions at the near- and far-field planes, which allows for the quantification of the non-Gaussianity of this state. This measure of non-Gaussianity requires only the measurement of the diagonal covariance sub-matrices, and will be relevant for new applications of the spatial correlation of SPDC in CV quantum information processing.
\end{abstract}

\ocis{(270.0270) Quantum optics; (270.5585) Quantum information and processing.}

\section{Introduction}
In spontaneous parametric down-conversion (SPDC), photon pairs are
generated with several degrees of freedom quantum correlated. In particular, the down-converted photons are spatially entangled. Due to energy and momentum conservation, the sum of the transverse momenta and the difference of the transverse positions
of the photons can be well defined, even though the position and
momentum of each photon are undefined \cite{BoydEPR,ShiEPR}. This type of Einstein-Podolsky-Rosen (EPR) correlation \cite{EPR}, has been used as a resource for fundamental studies of quantum mechanics \cite{BoydEPR,ShiEPR,FonsecaBROGLIE,Yarnall,PR}, for quantum imaging \cite{Pittman,Gatti}, and experiments of
quantum information \cite{LeoGen,BoydQu,BoydORBITAL}.

The process of SPDC has been studied extensively in the past \cite{Mandel,Klyshko,Monken1}, and much of the recent effort has been put in the quantification of the spatial entanglement for a given experimental geometry. Traditionally, this has been done through the technique of Schmidt decomposition of the two-photon wave function, which gives the Schmidt number, $K$, and the Schmidt modes allowed for each photon \cite{Eberly,Monken2,Torres,Kulik2}. While this approach describes important properties of the spatial correlation, it basically gives no information of the form of the spatial joint distributions, therefore giving no information about the Gaussianity of the spatial two-photon state of SPDC.

Moreover, it is well known that the state of the down-converted photons depends on the phase-matching conditions. Nevertheless, due to its complex structure, the phase-matching functions are usually approximated by Gaussian functions \cite{Torres,Kulik2,Ether,Paulao1,Paulao2,LeoCAM}. While the approach being adopted seems to be reasonable, it has already been shown that fine (new) structures in the spatial distributions of these photons can be observed due to (the manipulation of) the phase-matching conditions \cite{ExterNF1}.

In this work we study the non-Gaussianity of the spatial two-photon state of SPDC by properly taking into account the phase-matching conditions. We start by showing how the variances of the near- and far-field conditional probability distributions are affected by the phase-matching functions. Then, we analyze the role of the EPR-criterion \cite{EPR,Reid} regarding the non-Gaussianity and entanglement detection of this state. Even though it has been demonstrated that higher order separability criteria can be used for the entanglement detection of spatial non-Gaussian entangled states \cite{SteveNG}, we show that a proper consideration of the phase-matching function reveals, precisely, when the \emph{simpler} EPR-criterion can still be used for the spatial entanglement detection. We also show that (and when) the EPR-criterion can be used as a witness for the non-Gaussianity of this state.

Furthermore, we introduce a statistical measure, based on the \emph{negentropy} \cite{Negentropy} of the near- and far-field joint distributions, which allows for the quantification of the non-Gaussianity of the spatial two-photon state of SPDC. This measure does not correspond to a quantum mechanical generalization of the negentropy, such as the non-Gaussianity measure based on the quantum relative entropy (QRE) \cite{Banaszek,Paris}, and so does not require the knowledge of the full density matrix. Only the moments associated with the diagonal sub-matrices of the covariance matrix need to be measured. Thus, it is experimentally more accessible \cite{Eisert}. Moreover, for most of the configurations used so far, we show that its value can be estimated from the (easier to measure) marginal and conditional distributions. We also demonstrate that it has common properties with previous introduced measures of non-Gaussianity \cite{Banaszek,Paris}. The quantification of the non-Gaussianity of a quantum state has important applications for quantum information \cite{SteveNG,Banaszek,Paris}, and thus the practicality our measure shall be relevant for new applications of the spatial correlations of SPDC in this field.

\section{The phase-matching conditions and the variances of the conditional probabilities}
We consider the process of quasi-monochromatic SPDC, in the paraxial regime, for configurations with negligible Poynting vector transverse walk-off, that can be obtained using non-critical phase-matching techniques \cite{ExterNF1}.  In the momentum representation, the two-photon state is given by \cite{Mandel,Klyshko,Monken1}
\begin{equation} \label{EstMon}
|\Psi\rangle\propto \int\!\!\!\int d\brm{q}_1 d\brm{q}_2\ e^{-\frac{\zeta}{4}|\brm{q}_1+\brm{q}_2|^2} \sinc \frac{L|\brm{q}_1  -\brm{q}_2|^2}{4k_p} \ket{\brm{q}_1,\brm{q}_2},
\end{equation}
where $\ket{\brm{q}_1,\brm{q}_2}$ represents a two-photon state in plane-wave modes whose transverse wave vectors are $\brm{q}_1$ and $\brm{q}_2$.  $L$ is the crystal length, $k_p$ is the pump beam wave number, $\zeta=w_0^2-2iz/k_p$, and $w_0$ is the pump beam waist, which is located at  $z=0$. This state may be rewritten in the coordinate space as \cite{ExterNF1,ExterNF2}
\begin{equation} \label{EstPos}
|\Psi\rangle\propto\int\!\!\!\int d\boldsymbol{\rho}_1 d\boldsymbol{\rho}_2\ e^{-\frac{1}{4\zeta}|\brm{\rho}_1+\brm{\rho}_2|^2} \sint \frac{k_p|\brm{\rho}_1  -\brm{\rho}_2|^2}{4L}  \ket{\brm{\rho}_1,\brm{\rho}_2},
\end{equation} where we define the function $\sint$ as $\sint(x)\equiv \frac{2}{\pi}\int_x^{\infty} dt\,\sinc t\equiv 1-\frac{2}{\pi}\Si(x)$, $\Si(x)$ being the sine integral function. The functions $\sinc(b\,|\brm{q}|^2/2)$ and $\frac{\pi}{2b}\sint(|\brm{\rho}|^2/2b)$ form a Fourier transform pair.

The $\sinc$ and $\sint$ functions arise from the phase-matching conditions, and due to the difficulty of dealing analytically with them, they are usually approximated by Gaussian functions. The approach that has been adopted consists in approximating the function $\sinc(bx^2)$ by $e^{-\alpha b x^2}$. Sometimes it is used that $\alpha=1$ \cite{Eberly,Ether}, and in other cases the value of $\alpha$ is chosen such that both functions coincide at $1/e^2$ \cite{Torres,Paulao1,Paulao2} or at $1/e$ \cite{LeoCAM} of their peak. While this approximation seems to be reasonable for the entanglement quantification \cite{Eberly}, there has been no investigation to determine how precise it is for describing the distribution of the momentum correlations. Besides, it should be noticed that once a Gaussian approximation is adopted for the $\sinc$ function, the corresponding approximation for the $\sint$ function is already defined by the Fourier transform. Therefore, it is also not clear that such approximation is indeed good for describing the position distributions. Thus, it is not clear whether the SPDC two-photon state can indeed be written as a two-mode Gaussian state.

\begin{figure}[t]
\centering
\includegraphics[width=0.7\textwidth]{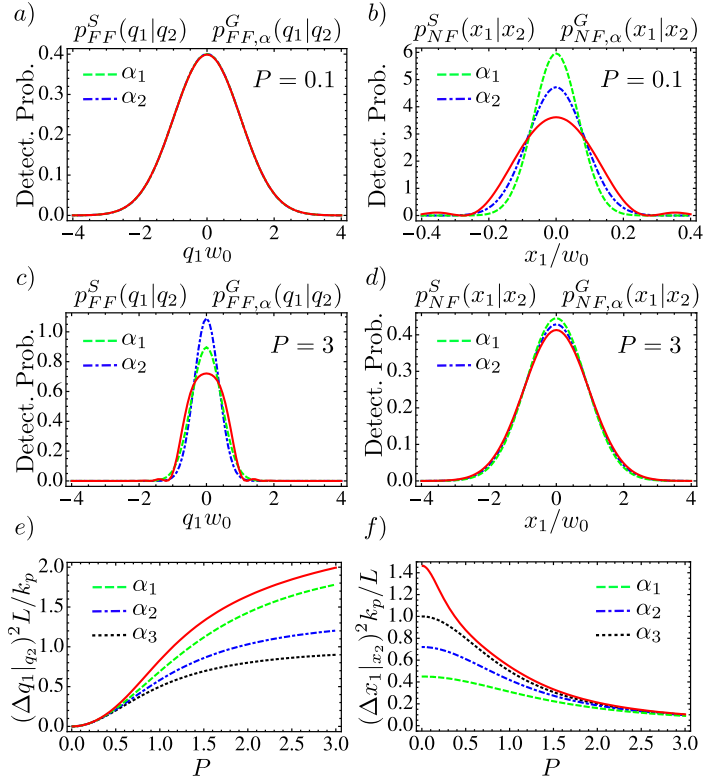}
\caption{The far- and near-field conditional probability distributions, while considering the state of SPDC ($p_{FF}^{S}$ and $p_{NF}^{S}$) and when some Gaussian approximations are assumed ($p_{FF,\alpha_{i}}^{G}$ and $p_{NF,\alpha_{i}}^{G}$). In ($a$)-($d$) the curves $p_{FF}^{S}$ and $p_{NF}^{S}$ (red solid lines) are compared with $p_{FF,\alpha_{2}}^{G}$ and $p_{NF,\alpha_{2}}^{G}$ for $\alpha_1=0.45$ (green dashed line), $\alpha_2=0.72$ (blue dot-dashed line), and for some specific values of $P$. In ($e$) and ($f$) the variances of the momentum and position conditional probabilities are plotted in terms of $P$, respectively. $\alpha_3=1$.} \label{Fig1}
\end{figure}

To investigate this point we study how the variances of the momentum (far-field) and position (near-field) conditional probability distributions are affected by the phase-matching function, and compare the obtained results with the cases where Gaussian approximations are considered. Due to the symmetry of the two-photon wave functions, there is no loss of generalization if we work in one dimension (i.e., $y_1=y_2=0$ and $q_{y1}=q_{y2}=0$). To simplify our analysis we define the following dimensionless variables: $\tilde{x}_j={x}_j/w_0$ and $\tilde{q}_j=w_0{q}_{j}$, $j=1,2$. The probabilities for coincidence detection at the far- and near-field planes are
\begin{eqnarray}
p_{FF}^S({\tilde{q}}_1,{\tilde{q}}_2)&\propto& e^{-\frac{1}{2}({\tilde{q}}_1+{\tilde{q}}_2)^2}\, \left(\sinc \frac{P^2({\tilde{q}}_1 -{\tilde{q}}_2)^2}{4}\right)^2 ,\label{pff}\\
p_{NF}^S({\tilde{x}}_1,{\tilde{x}}_2)&\propto& e^{-\frac{1}{2\sigma^2}({\tilde{x}}_1+{\tilde{x}}_2)^2}\left(\sint\frac{({\tilde{x}}_1-{\tilde{x}}_2)^2}{4P^2}\right)^2\!,\label{pnf}
\end{eqnarray} where  $\sigma^2=(z_0^2+z^2)/z_0^2$, $P=\sqrt{L/(2z_0)}$, and $z_0=k_pw_0^2/2$ is the diffraction length of the pump beam. The joint probabilities are related with the conditional (and marginal) probabilities $p_{FF}^S({\tilde{q}}_1|{\tilde{q}}_2)$ [$p_{FF}^S({\tilde{q}}_j$)] and $p_{NF}^S({\tilde{x}}_1|{\tilde{x}}_2)$ [$p_{NF}^S({\tilde{x}}_j$)], through the rule: $p({\xi}_1|{\xi}_2)=\frac{p({\xi}_1,{\xi}_2)}{p({\xi}_2)}$.

In the case of the Gaussian approximations discussed above, the probabilities of coincidence detection are
\begin{eqnarray}
p_{FF,\alpha_{i}}^{G}({\tilde{q}}_1,{\tilde{q}}_2) &\propto& e^{-\frac{1}{2}({\tilde{q}}_1+{\tilde{q}}_2)^2}\,e^{-\alpha_{i}\frac{P^2}{2}({\tilde{q}}_1-{\tilde{q}}_2)^2},\\
p_{NF,\alpha_{i}}^{G}({\tilde{x}}_1,{\tilde{x}}_2) &\propto& e^{-\frac{1}{2\sigma^2}({\tilde{x}}_1+{\tilde{x}}_2)^2}e^{-\frac{1}{2\alpha_{i}P^2}({\tilde{x}}_1-{\tilde{x}}_2)^2},
\end{eqnarray} where different values of $\alpha_i$ represent distinct Gaussian approximations for the $\sinc$ function. In Figure \ref{Fig1}(a) [(b)] we compare the curves $p_{FF}^{S}$ and $p_{FF,\alpha_{i}}^{G}$ ($p_{NF}^{S}$ and $p_{NF,\alpha_{i}}^{G}$) considering $\tilde{x}_2,\tilde{q}_2=0$, $\sigma=1$ (and the crystal centered at $z=0$), for the case where the dimensionless parameter $P=0.1$. This parameter has been used in the study of the quantification of the spatial entanglement \cite{Eberly,Monken2}, and it brings universality to the theory since a certain value of $P$ can be reached in three different ways. Here we considered the values of $L$ and $k_p$ as fixed parameters such that $P$ varies with $w_0$. In this case we find $\alpha_1=0.45$ ($\alpha_2 =0.72$) for the case where the $\sinc$ and Gaussian functions coincide at $1/e$ ($1/e^2$). From Figure \ref{Fig1}(a)-(b) one can see that for a small value of $P$, the Gaussian approximation only describes properly the momentum conditional distribution. The position conditional distribution is barely described by the approximation. In Figure \ref{Fig1}(c) and Figure \ref{Fig1}(d) we have the same type of analysis but now for a larger value of $P$. In this case, the Gaussian approximation is useful only for describing the position conditional distribution. The overall behavior of the Gaussian approximations is showed in Figure \ref{Fig1}(e) and Figure \ref{Fig1}(f), where the normalized variances of the far- [$(\Delta q_1|_{q_2})^2 L/k_p$] and near-field [$(\Delta x_1|_{x_2})^2 k_p/L$] conditional distributions are plotted in terms of $P$.

\section{The EPR-Criterion as a witness for the non-Gaussianity of the spatial two-photon state}
The near and far-field conditional probabilities can be used for implementing the EPR-paradox \cite{BoydEPR,EPR,Paulao2,Reid}. This is done by observing the violation of the inequality $(\Delta x_1|_{x_2})^2 (\Delta q_1|_{q_2})^2 \geq \frac{1}{4}$, which certifies the quantum nature of the spatial correlations of the down-converted photons. Since we have determined how the phase-matching function affects the variances $(\Delta x_1|_{x_2})^2$ and $(\Delta q_1|_{q_2})^2$, we can also look for its effect on the EPR-criterion. This is showed with the red (solid) line in Figure \ref{Fig2}, which was calculated for the values of $\tilde{x}_2$ and $\tilde{q}_2$ at the origin. For smaller values of $P$, the conditional variances are independent of the $\tilde{x}_2$ and $\tilde{q}_2$ values \cite{BoydEPR}. Whenever $P$ increases, the variances become dependent on $\tilde{x}_2$ and $\tilde{q}_2$. Nevertheless for smaller values of $\tilde{x}_2$ and $\tilde{q}_2$, which are of most experimental relevance, the red (solid) curve shown in Figure \ref{Fig2} captures the overall behavior of the product of $(\Delta x_1|_{x_2})^2$ and $(\Delta q_1|_{q_2})^2$ for the state of Eq~(\ref{EstMon}). As we can see, for values of $P$ smaller than $0.56$ or greater than $2.58$, the EPR-criterion can safely be used for detecting the spatial entanglement of the two-photon state of SPDC.

Besides of being useful as a entanglement witness, the EPR-criterion can also be used as a witness for the non-Gaussianity of the spatial state of SPDC. This emphasizes another application for this criterion, which has been related already with other quantum information tasks \cite{Reid}. To observe this, we note that for a pure two-mode Gaussian state it is possible to show that $(\Delta x_1|_{x_2})^2 (\Delta q_1|_{q_2})^2 = \frac{1}{4K}$ [see Appendix \ref{SuplementoA}], such that $\frac{1}{4}$ is an upper bound for the EPR-criterion with these states. Since it has been demonstrated in \cite{Eberly} that the two-photon state of Eq.~(\ref{EstMon}) is always entangled, we can say whenever the product of variances is greater than $\frac{1}{4}$, that it witnesses the non-Gaussianity of the entangled spatial two-photon state of SPDC. As one can see in Figure \ref{Fig2}, this happens for $0.56 \leq P \leq 2.58$. The Schmidt decomposition of the state wave function, used together with the EPR-criterion value, reveals the non-Gaussianity of a two-mode entangled state. Such observation does not necessarily hold true when other second-order moments criteria are considered. This is shown in Appendix \ref{SuplementoB} for the criterion of Ref. \cite{Mancini}, and for the states considered in Figure \ref{Fig2}.

\begin{figure}[h]
\centering
\includegraphics[width=0.7\textwidth]{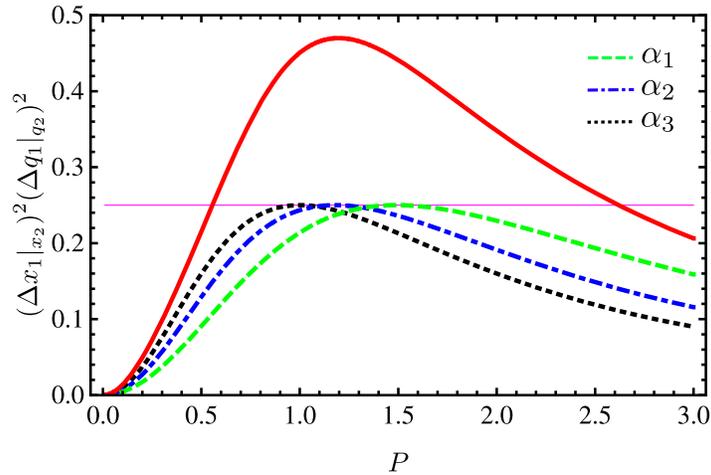}
\caption{The EPR-criterion plotted in terms of $P$. The red (solid) curve corresponds to the EPR values of the two-photon state generated in the SPDC [Eq.~(\ref{EstMon})]. The green (dotted), blue (dot-dashed) and black (dotted) curves describe the EPR-criterion for the two-photon Gaussian states defined in terms of $\alpha_1$, $\alpha_2$ and $\alpha_3$, respectively.} \label{Fig2}
\end{figure}
\section[Quantifying the non-Gaussianity of the spatial two-photon state...]{Quantifying the non-Gaussianity of the spatial two-photon state of SPDC}
From our previous analysis it is clear that the spatial state of SPDC can not be seen as a two-mode Gaussian entangled state, even when it is generated with a Gaussian pump beam and in the case of perfect phase-matching. We now proceed to quantifying the non-Gaussianity of this state. First, we introduce the concept of \emph{negentropy} which is the base of our approach \cite{Negentropy}. The negentropy of a probability density function $p(\xi_1,\xi_2)$ is defined as $ N\equiv H[p^{\tilde{G}}(\xi_1,\xi_2)]-H[p(\xi_1,\xi_2)],$ where $p^{\tilde{G}}(\xi_1,\xi_2)$ is a Gaussian distribution with the same expected values and covariance matrix of $p(\xi_1,\xi_2)$. The function $H[p(\xi_1,\xi_2)]$, called \emph{differential entropy}, is defined as $H[p(\xi_1,\xi_2)]\equiv-\!\int\! d\xi_1 d\xi_2\,p(\xi_1,\xi_2)\log_2p(\xi_1,\xi_2)$ \cite{Negentropy2}. The advantage of using negentropy is that it can be seen as the optimal estimator of non-Gaussianity, as far as density probabilities are involved. This is due to the properties that it is always non-negative, and that it is zero only for Gaussian distributions. Besides, it is invariant under invertible linear transformations \cite{Comon}.

Motivated by these properties we define the total non-Gaussianity of the spatial two-photon state of SPDC as
\begin{equation} \label{ngT}
nG^T\equiv N[p_{FF}^S({\tilde{q}}_1,{\tilde{q}}_2)]+ N[p_{NF}^S({\tilde{x}}_1,{\tilde{x}}_2)],
\end{equation} where $N[p_{FF}^S({\tilde{q}}_1,{\tilde{q}}_2)]$ and $N[p_{NF}^S({\tilde{x}}_1,{\tilde{x}}_2)]$ are the negentropies of the far- and near-field joint distributions of Eq.~(\ref{pff}) and Eq.~(\ref{pnf}), respectively. In Appendix \ref{SuplementoC} we give, explicitly, the calculations of $N[p_{FF}^S({\tilde{q}}_1,{\tilde{q}}_2)]$ and $N[p_{NF}^S({\tilde{x}}_1,{\tilde{x}}_2)]$. It is possible to observe that $nG^T=0$ if and only if $|\Psi\rangle$ is a two-mode Gaussian state [see Appendix \ref{SuplementoD}]. Otherwise $nG^T>0$. We obtain that $N[p_{FF}^S({\tilde{q}}_1,{\tilde{q}}_2)] \approx 0.15$ and that $N[p_{NF}^S({\tilde{x}}_1,{\tilde{x}}_2)] \approx 0.22$. Thus, the total non-Gaussianity of the spatial two-photon state of SPDC [Eq~(\ref{EstMon})] is $nG^T\approx0.37$. It is interesting to note that it does not depend on $P$. This was expected since the phase-matching functions do not change their functional form when $P$ varies. The calculations performed for these negentropies can be adapted for different experimental geometries, or used for the proper determination of other quantum information quantities related with the differential entropy of the spatial joint distributions \cite{Howell}. The fact that $nG^T\neq0$, emphasizes that spatial Gaussian approximations should be taken carefully due to the \textit{extremality} of Gaussian states \cite{Cirac}. For comparison purposes, we calculated in Appendix \ref{SuplementoE} the value of $\delta_B$, which is the measure of non-Gaussianity based on the QRE \cite{Banaszek,Paris}. We obtain that $\delta_B=1.08$ and such result also highlight the non-Gaussian character of the state of Eq.~(\ref{EstMon}). Furthermore, it has the same behavior of $nG^T$, since it does not depend on the parameter $P$.

In analogy with Eq.~(\ref{ngT}), we define the non-Gaussianity of the conditional and marginal distributions as: $nG^C\equiv N[p^S_{FF}(q_i|q_j)]+N[p^S_{NF}(x_i|x_j)]$ and $nG^M \equiv N[p^S_{FF}(q_i)]+N[p^S_{NF}(x_i)]$ with $i,j=1,2$ and $i \neq j$. According to these definitions, one can observe that the non-Gaussianity of the spatial state of SPDC decreases under partial trace, such that $nG^T > nG^M$; and that it is additive when the composite system is represented by a product state, i.e., if $\ket{\Psi}$ is a product state, then $nG^T=2 nG^M$ [see Appendix \ref{SuplementoD}]. These are common properties with the QRE measure of \cite{Banaszek,Paris}.

In Figure \ref{Fig3}(a) we plot the negentropies $N[p^S_{FF}(q_1|q_2)]$ and $N[p^S_{NF}(x_1|x_2)]$ as a function of $P$. It is interesting to note that these curves quantify the idea already presented in Figure \ref{Fig1}(e)-(f). As larger the value of $P$ is, the less the conditional momentum distribution can be approximated by a Gaussian function. On the other hand, the conditional position probabilities tends to a normal distribution when $P$ increases. In Figure \ref{Fig3}(b) we plot the negentropies of the near- and far-field marginal probabilities. One can see that they have a different dependence on $P$ in comparison with the conditional probabilities. Now, the near-field distribution tends to a Gaussian function for larger values of $P$, and the far-field one for smaller values of $P$. In Figure \ref{Fig3}(c) and Figure \ref{Fig3}(d) we have $nG^C$ and $nG^M$ plotted in terms of $P$. The insets of these figures show the corresponding near- and far-field distributions at the points of minimum, indicated with red circles.

\begin{figure}[t]
\centering
\includegraphics[angle=-90,width=0.8\textwidth]{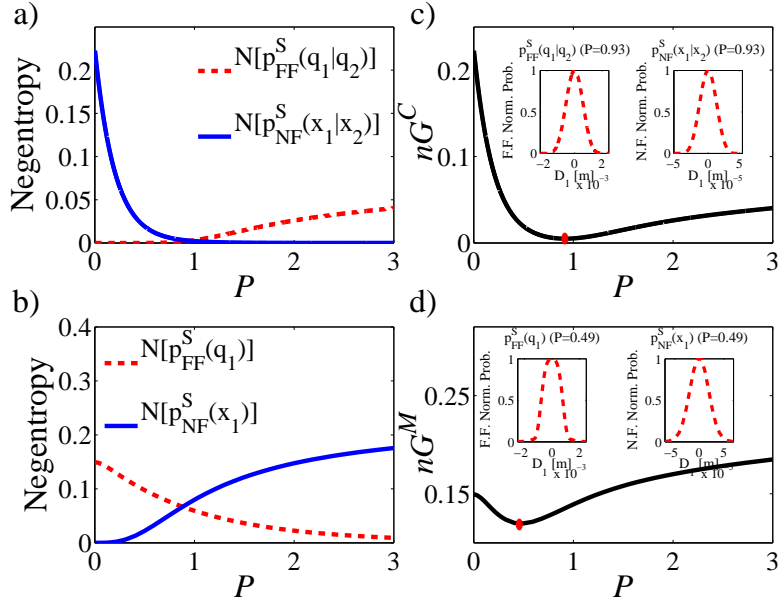}
\caption{In (a)[(b)] the negentropies of the near- and far-field conditional (marginal) distributions are plotted in terms of $P$. In (c) [(d)] we show the non-Gaussianity of the conditional (marginal) distributions.} \label{Fig3}
\end{figure}

As it is shown in Appendix \ref{SuplementoF} [See also Figure \ref{Fig3}(c) and Figure \ref{Fig3}(d)], in the limit of $P\ll1$, one can use the relation between the joint and conditional density probabilities to decompose $nG^T$ as the sum of $nG^C$ and $nG^M$:
\begin{equation} \label{ngt2}
nG^T \approx nG^C+nG^M \ \ (\textrm{iff}\ P \ll 1).
\end{equation} In a typical experimental configuration for SPDC, where the pump beam spot size is around $1$~mm at the crystal plane, the value of $P$ can be smaller than $0.05$. Thus, in general, the total non-Gaussianity of the spatial two-photon state of SPDC can be estimated in terms of the (easier to measure) near- and far-field negentropies of the conditional and marginal distributions. This simplify the measurement of $nG^T$, since there is no need to scan the whole transverse planes associated with the near- and far-field joint distributions.

\section{Conclusion}
We have investigated the spatial distributions of the entangled down-converted photons by proper considering the phase-matching conditions. By understanding how the near- and far-field conditional distributions are affected by the phase-matching function, we could show that the EPR-criterion \cite{EPR} can be used as a witness for the non-Gaussianity of the spatial state of the SPDC. The work culminated in the quantification of the non-Gaussianity of this state, which was based in a new and very experimentally accessible measure. As it has been discussed in \cite{Banaszek,Paris}, the quantification of the non-Gaussianity of a quantum state has many applications in the area of continuous variables quantum information processing. Thus, we envisage the use of our result for new applications of the spatial correlations of SPDC.

\appendix
\section{EPR-criterion for spatial Gaussian two-photon states\label{SuplementoA}}
As we mentioned in the main paper, the momentum representation of the spatial two-photon state generated in the spontaneous parametric down-conversion (SPDC) process, under paraxial approximation and for configurations with negligible Poynting vector transverse walk-off, can be written as  \cite{Mandel,Klyshko,Monken1}
\begin{equation}
|\Psi\rangle 
= \int dq_1\int dq_2\,\Psi(q_1,q_2) |1q_1\rangle|1q_2\rangle,
\end{equation}
where the amplitude $\Psi(q_1,q_2)$ is given by
\begin{equation}
\Psi(q_1,q_2) \propto \Delta(q_1+q_2)\Theta(q_1-q_2),
\end{equation}
with $\Delta(q_1+q_2)$ being the angular spectrum of the pump beam and $\Theta(q_1-q_2)$ representing the phase-matching conditions of the non-linear process. If both functions are represented by Gaussian functions of the form \hbox{$\Delta(q_1+q_2)=\exp\left[-\frac{(q_1+q_2)^2}{2\sigma_+^2}\right]$} and \hbox{$\Theta(q_1-q_2)=\exp\left[-\frac{(q_1-q_2)^2}{2\delta_-^2}\right]$} \cite{Eberly,Torres,Kulik2,Ether,Paulao1,Paulao2,LeoCAM}, we have in momentum representation that
\begin{equation}\label{gaussian-q}
\Psi(q_1,q_2)\propto e^{-(q_1+q_2)^2/2\sigma_+^2}e^{-(q_1-q_2)^2/2\delta_-^2},
\end{equation}
and in position representation that
\begin{equation}\label{gaussian-x}
\Psi(x_1,x_2)\propto e^{-\sigma_+^2(x_1+x_2)^2/8}e^{-\delta_-^2(x_1-x_2)^2/8}.
\end{equation}
The parameter $\delta_-$ can be adjusted in order to approximate the phase-matching function by distinct Gaussian functions. For a pump laser beam with a Gaussian transverse profile, $\sigma_+^2=2/c^2$, where $c$ is the radius of this pump at the plane of the non-linear crystal.

Based on Eqs. (\ref{gaussian-q}) and (\ref{gaussian-x}), we can obtain the probability density functions for the conditional position and momentum distributions of the down-converted photons. The variances of these curves are, in general, independent of the value considered for $x_2$ and $q_2$ \cite{BoydEPR} and here we use, for simplicity, $x_2\!=\!0$ and $q_2\!=\!0$. In this case $P(x_1|x_2\!=\!0)$ and $P(q_1|q_2\!=\!0)$ are given by
\begin{equation}\label{P-x}
P(x_1|x_2\!=\!0)=\frac{\sqrt{\sigma_+^2+\delta_-^2}}{2\sqrt{\pi}}\exp\left[-\frac{\sigma_+^2+\delta_-^2}{4}x_1^2\right],
\end{equation} and
\begin{equation}\label{P-q}
P(q_1|q_2\!=\!0)=\frac{\sqrt{\sigma_+^2+\delta_-^2}}{\sqrt{\pi}\sigma_+^2\delta_-^2}\exp\left[-\frac{\sigma_+^2+\delta_-^2}{\sigma_+^2 \delta_-^2} q_1^2\right].
\end{equation}
Here $\sigma_+^2$ and $\delta_-^2$ are the widths of the respective Gaussian functions \cite{Eberly}. The variances of the conditional distributions can be calculated directly from Eqs. (\ref{P-x}) and (\ref{P-q}), which results in
\begin{equation}\label{D-x}
(\Delta x_1|_{x_2})^2 = \frac{2}{\sigma_+^2+\delta_-^2},
\end{equation} and
\begin{equation}\label{D-q}
(\Delta q_1|_{q_2})^2 = \frac{\sigma_+^2\delta_-^2}{2(\sigma_+^2+\delta_-^2)}.
\end{equation}
Then, the product of them is given by \vspace{-0.2cm}
\begin{equation}\label{prodDxDq}
(\Delta x_1|_{x_2})^2 (\Delta q_1|_{q_2})^2 = \frac{\sigma_+^2\delta_-^2}{(\sigma_+^2+\delta_-^2)^2}.
\end{equation}

Let us now consider the Schmidt number $K$ for Gaussian states. It was showed in Ref. \cite{Eberly} that it is given by the expression
\begin{equation}\label{Schmidt}
K=\frac{1}{4}\left(\frac{\sigma_+}{\delta_-} + \frac{\delta_-}{\sigma_+}\right)^2=\frac{1}{4}\frac{(\sigma_+^2+\delta_-^2)^2}{\sigma_+^2\delta_-^2}.
\end{equation}
Comparing Eqs. (\ref{prodDxDq}) and (\ref{Schmidt}), it follows immediately that the product of the conditional variances is a function of the Schmidt number, such that
\begin{equation}
(\Delta x_1|_{x_2})^2(\Delta q_1|_{q_2})^2 = \frac{1}{4K}.
\end{equation}

This expression is valid for any Gaussian approximation taken for the phase-matching function. It is possible to see that the maximal value of $(\Delta x_1|_{x_2})^2 (\Delta q_1|_{q_2})^2$ is equal to $\frac{1}{4}$ and that it happens when the Schmidt number $K=1$ (i.e., for product Gaussian states).

\section{Mancini \emph{et al.} Criterion for the spatial entanglement of SPDC\label{SuplementoB}}
Another entanglement detection criterion, based on second-order moments, is the one introduced by Mancini \emph{et al.} \cite{Mancini}. For the case of spatial entanglement it reads \cite{Eisert}
\begin{equation}\label{mancini}
[\Delta(\tilde{x}_1-\tilde{x}_2)]^2[\Delta(\tilde{q}_1+\tilde{q}_2)]^2\geq 1.
\end{equation}
When this inequality is violated, the state is not separable. If this condition is satisfied, no information can be drawn. For the far and near-field joint spatial distributions [Eq.~(3) and Eq.~(4) of main paper, respectively], the variances in (\ref{mancini}) are given by
\begin{equation}
[\Delta(\tilde{q}_1+\tilde{q}_2)]^2 \propto \int d\tilde{q}_+\, \tilde{q}_+^2 e^{-\frac{1}{2}\tilde{q}_+^2}=1,
\end{equation}
\begin{equation}
[\Delta(\tilde{x}_1-\tilde{x}_2)]^2 \propto \int d\tilde{x}_-\, \tilde{x}_-^2 \sint^2\frac{1}{4P}\tilde{x}_-^2=4\frac{A_2}{A_1}P^2,
\end{equation}
where $\tilde{q}_+=\tilde{q}_1+\tilde{q}_2$ and $\tilde{x}_-=\tilde{x}_1-\tilde{x}_2$. The constants
\begin{equation}\label{A1}
A_1=\int_{-\infty}^\infty d\xi\,\sint^2\xi^2\approx 1.4008,
\end{equation}
and
\begin{equation}\label{A2}
A_2=\int_{-\infty}^\infty d\xi\,\xi^2\sint^2\xi^2\approx 0.5897.
\end{equation}

Figure \ref{FigMancini} shows the dependence of the Mancini-Criterion with the parameter $P$, while considering the state of SPDC and when distinct Gaussian approximations for this state are considered. Since there is no upper limit for this criterion with Gaussian two-mode states, it is not possible to use the Mancini-Criterion for the detection of the non-Gaussianity of the spatial state of SPDC.
\begin{figure}[h]
\centering
\includegraphics[angle=-90,width=0.6\textwidth]{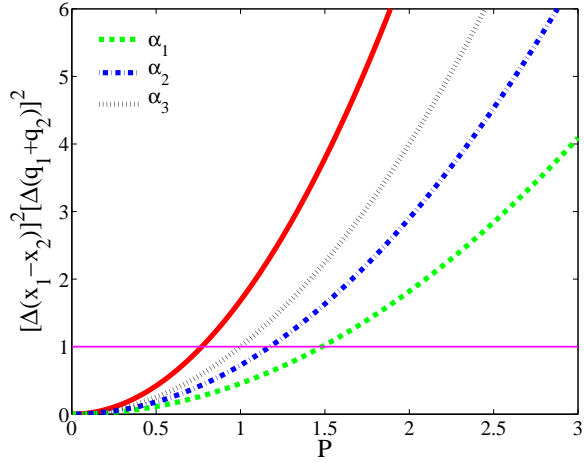}
\caption{Comparison of the Mancini-criterion while considering the state of SPDC (red solid line) and when some Gaussian approximations are adopted. Here we use $\alpha_1=0.45$ (green dashed line), $\alpha_2=0.72$ (blue dot-dashed line), and $\alpha_3=1$ (black dotted line) to describe the spatial Gaussian approximations [See Eq.~(5) and Eq.~(6) of main paper].} \label{FigMancini}
\end{figure}

\section{Negentropy of far- and near-field joint distributions of SPDC\label{SuplementoC}}
\subsection{Negentropy of a probability density function}
Let $p(\xi_1,\xi_2)$ be a probability density function. The Negentropy ($N$) of $p(\xi_1,\xi_2)$ is defined as \cite{Negentropy}
\begin{equation}\label{negent}
N\equiv H[p^{\tilde{G}}(\xi_1,\xi_2)]-H[p(\xi_1,\xi_2)],
\end{equation}
where $p^{\tilde{G}}(\xi_1,\xi_2)$ is a Gaussian distribution with the same expected value vector $\boldsymbol{\mu}=\{\mu_{\xi_1},\mu_{\xi_2}\}$ and same covariance matrix $\Lambda$ of $p(\xi_1,\xi_2)$. The function $H[p(\xi_1,\xi_2)]$, called \emph{differential entropy}, is defined as \cite{Negentropy2}
\begin{equation}\label{diffent}
H[p(\xi_1,\xi_2)]\equiv-\!\int\! d\xi_1\!\int\!d\xi_2\,p(\xi_1,\xi_2)\log_2(p(\xi_1,\xi_2)).
\end{equation}

\subsection{Negentropy of far-field joint distribution}
The joint probability density function in the far-field plane is given by
\begin{equation}
p^S_{FF}(\tilde{q}_1,\tilde{q}_2)=\frac{3P}{4\sqrt{2}\pi}e^{-\frac{(\tilde{q}_1+\tilde{q}_2)^2}{2}}\sinc^2\frac{P^2}{4}(\tilde{q}_1-\tilde{q}_2)^2.\label{psff}\end{equation}

The covariance matrix of $p^S_{FF}(\tilde{q}_1,\tilde{q}_2)$ is given by
\begin{equation}
\Lambda^S_{FF}\!=\!
\left(\begin{array}{cc}
(\Delta \tilde{q}_1)^2_S & \langle \tilde{q}_1\tilde{q}_2\rangle_S -\langle \tilde{q}_1\rangle_S\langle \tilde{q}_2\rangle_S\\
\langle \tilde{q}_2\tilde{q}_1\rangle_S -\langle \tilde{q}_2\rangle_S\langle \tilde{q}_1\rangle_S & (\Delta \tilde{q}_2)^2_S
\end{array}\right).
\end{equation}
The elements in the diagonal are the variances of the marginal distributions $p^S_{FF}(\tilde{q}_1)$ and $p^S_{FF}(\tilde{q}_2)$. In order to calculate the terms of $\Lambda^S_{FF}$, we first note that the expected values $\boldsymbol{\mu}_{FF}=\{\langle \tilde{q}_1\rangle,\langle \tilde{q}_2\rangle\}$ are null. So, we have that
\begin{eqnarray}
(\Delta \tilde{q}_1)^2_S=(\Delta \tilde{q}_2)^2_S&=&\int\!\!\!\int d\tilde{q}_1 d\tilde{q}_2\,\tilde{q}_1^2\,p^S_{FF}(\tilde{q}_1,\tilde{q}_2)\nonumber\\
&=&\frac{1}{4}\left(1+\frac{3}{P^2}\right)\label{varff},
\end{eqnarray}
and
\begin{eqnarray}
\langle \tilde{q}_1\tilde{q}_2\rangle_S=\langle \tilde{q}_2\tilde{q}_1\rangle_S&=&\int\!\!\!\int d\tilde{q}_1 d\tilde{q}_2\,\tilde{q}_1\tilde{q}_2\,p^S_{FF}(\tilde{q}_1,\tilde{q}_2)\nonumber\\
&=&\frac{1}{4}\left(1-\frac{3}{P^2}\right)\label{covarff}.
\end{eqnarray}

Let us consider the following Gaussian approximation for describing the joint probability density distribution at the far-field plane:
\begin{equation}
p^{G}_{FF}(\tilde{q}_1,\tilde{q}_2)=\frac{1}{\pi\delta_-}e^{-\frac{(\tilde{q}_1+\tilde{q}_2)^2}{2}}e^{-\frac{(\tilde{q}_1-\tilde{q}_2)^2}{2\delta_-^2}}.\label{pgff}
\end{equation}
The elements of the covariance matrix $\Lambda^{G}_{FF}$ of $p^{G}_{FF}(\tilde{q}_1,\tilde{q}_2)$ are given by

\begin{eqnarray}
(\Delta \tilde{q}_1)^2_{G}=(\Delta \tilde{q}_2)^2_{G}&=&\int\!\!\!\int d\tilde{q}_1 d\tilde{q}_2\,\tilde{q}_1^2\,p^{G}_{FF}(\tilde{q}_1,\tilde{q}_2)\nonumber\\
&=&\frac{1}{4}\left(1+\delta_-^2\right)\label{varffgaus},
\end{eqnarray}
and
\begin{eqnarray}
\langle \tilde{q}_1\tilde{q}_2\rangle_{G}=\langle \tilde{q}_2\tilde{q}_1\rangle_{G}&=&\int\!\!\!\int d\tilde{q}_1 d\tilde{q}_2\,\tilde{q}_1\tilde{q}_2\,p^{G}_{FF}(\tilde{q}_1,\tilde{q}_2)\nonumber\\
&=&\frac{1}{4}\left(1-\delta_-^2\right)\label{covarffgaus}.
\end{eqnarray}

From Eq. (\ref{varff}), (\ref{covarff}), (\ref{varffgaus}) and (\ref{covarffgaus}) we can observe that the condition
\begin{equation}
\delta_-^2=\frac{3}{P^2},
\end{equation} implies that the joint distributions (\ref{psff}) and (\ref{pgff}) have the same covariance matrix. Hereafter, the Gaussian distribution with the same covariance matrix of (\ref{psff}) is denoted by $p^{\tilde G}_{FF}(\tilde{q_1},\tilde{q_2})$.

It is possible to show that the differential entropy of the Gaussian distribution $p^{\tilde G}_{FF}(\tilde{q}_1,\tilde{q}_2)$ is
\begin{equation}
H[p^{\tilde G}_{FF}(\tilde{q}_1,\tilde{q}_2)]=\log_2\frac{\pi e\sqrt{3}}{P}.
\end{equation}

The differential entropy of $p^S_{FF}(\tilde{q}_1,\tilde{q}_2)$ is given by
\begin{eqnarray}
H[p^S_{FF}(\tilde{q}_1,\tilde{q}_2)]&=&-\int\!\!\!\int d\tilde{q}_1 d\tilde{q}_2\,p^S_{FF}(\tilde{q}_1,\tilde{q}_2)\log_2(p^S_{FF}(\tilde{q}_1,\tilde{q}_2))\nonumber\\
&=&\log_2\frac{4\sqrt{2}\pi}{3P}-\frac{3P}{4\sqrt{2}\pi}\int\!\!\!\int d\tilde{q}_1 d\tilde{q}_2\,\left\lbrace e^{-\frac{(\tilde{q}_1+\tilde{q}_2)^2}{2}}\right.\nonumber\\
& &\times\left.\left[\log_2e^{-\frac{(\tilde{q}_1+\tilde{q}_2)^2}{2}}+\log_2\sinc^2\frac{P^2}{4}(\tilde{q}_1-\tilde{q}_2)^2\right]\right\rbrace.
\end{eqnarray}

Making $u=\frac{\tilde{q}_1+\tilde{q}_2}{\sqrt{2}}$ and $v=\frac{P}{2}(\tilde{q}_1-\tilde{q}_2)$ we obtain that
\begin{eqnarray}
H[p^S_{FF}(\tilde{q}_1,\tilde{q}_2)]&=&\log_2\frac{4\sqrt{2}\pi}{3P}\frac{1}{2\ln 2}\nonumber\\
& &\times\left[1-\frac{3}{\sqrt{\pi}}\int_0^\infty\!\! dv\,\sinc^2 v^2\ln\left(\sinc^2 v^2\right)\right].\label{intnum}
\end{eqnarray}

The last integral in (\ref{intnum}) must be calculated numerically and it gives that
\begin{equation*}
\int_0^\infty\!\! dv\,\sinc^2 v^2\ln\left(\sinc^2 v^2\right)\approx-0.364.
\end{equation*}

Therefore, the differential entropy of $p^S_{FF}(\tilde{q}_1,\tilde{q}_2)$ is
\begin{equation}
H[p^S_{FF}(\tilde{q}_1,\tilde{q}_2)]\approx\log_2\frac{4\sqrt{2}\pi}{3P}+1.17.
\end{equation}

According to Eq. (\ref{negent}), we have that
\begin{eqnarray}
N[p^S(\tilde{q}_1,\tilde{q}_2)]&=&H[p^{\tilde{G}}_{FF}(\tilde{q}_1,\tilde{q}_2)]-H[p^S_{FF}(\tilde{q}_1,\tilde{q}_2)] \nonumber \\
& &\approx\log_2\frac{\pi e\sqrt{3}}{P}-\log_2\frac{4\sqrt{2}\pi}{3P}-1.17\nonumber\\
& &\approx\log_2 e\sqrt{\frac{27}{32}}-1.17\nonumber\\
& &\approx 0.15.
\end{eqnarray}

\subsection{Negentropy of near-field joint distribution}
The joint probability density function in the near-field plane is given by
\begin{equation}
p^S_{NF}(\tilde{x}_1,\tilde{x}_2)=\frac{1}{C}e^{-\frac{(\tilde{x}_1+\tilde{x}_2)^2}{2\sigma}}\sint^2\frac{(\tilde{x}_1-\tilde{x}_2)^2}{4P^2},\label{psnf}
\end{equation}
where $C$ is a normalization constant given by
\begin{equation}
C=\sqrt{2\pi}A_1\sqrt{\sigma}P.
\end{equation}

The covariance matrix for the near-field joint distribution is
\begin{equation}
\Lambda^S_{NF}=
\left(\begin{array}{cc}
(\Delta \tilde{x}_1)^2_S& \langle \tilde{x}_1\tilde{x}_2\rangle_S\\
\langle \tilde{x}_2\tilde{x}_1\rangle_S& (\Delta \tilde{x}_2)^2_S
\end{array}\right),
\end{equation}
where
\begin{eqnarray}
(\Delta \tilde{x}_1)^2_S=(\Delta \tilde{x}_2)^2_S&=&\int\!\!\!\int d\tilde{x}_1 d\tilde{x}_2\,\tilde{x}_1^2\,p^S_{NF}(\tilde{x}_1,\tilde{x}_2)\nonumber\\
&=&\frac{1}{4}\left(\sigma+\frac{4A_2 P^2}{A_1}\right)\label{varnf},
\end{eqnarray}
and
\begin{eqnarray}
\langle \tilde{x}_1\tilde{x}_2\rangle_S=\langle \tilde{x}_2\tilde{x}_1\rangle_S&=&\int\!\!\!\int d\tilde{x}_1 d\tilde{x}_2\,\tilde{x}_1\tilde{x}_2\,p^S_{NF}(\tilde{x}_1,\tilde{x}_2)\nonumber\\
&=&\frac{1}{4}\left(\sigma-\frac{4A_2 P^2}{A_1}\right)\label{covarnf}.
\end{eqnarray}

Let us now consider the following Gaussian approximation for the joint probability distribution at the near-field plane:
\begin{equation}
p^G_{NF}(\tilde{x}_1,\tilde{x}_2)=\frac{1}{\pi\sqrt{\sigma}\tau}e^{-\frac{(\tilde{x}_1+\tilde{x}_2)^2}{2\sigma}}e^{-\frac{(\tilde{x}_1-\tilde{x}_2)^2}{2\tau^2}}.\label{pgnf}
\end{equation}

The covariance matrix $\Lambda^G_{NF}$ of $p^G_{NF}(\tilde{x}_1,\tilde{x}_2)$ is determined by
\begin{equation}
\Lambda^G_{NF}=
\left(\begin{array}{cc}
(\Delta \tilde{x}_1)^2_G\!\!&\! \langle \tilde{x}_1\tilde{x}_2\rangle_G \\
\langle \tilde{x}_2\tilde{x}_1\rangle_G  \!&\! (\Delta \tilde{x}_2)^2_G
\end{array}\right),
\end{equation}
where
\begin{eqnarray}
(\Delta \tilde{x}_1)^2_{G}=(\Delta \tilde{x}_2)^2_{G}&=&\int\!\!\!\int d\tilde{x}_1 d\tilde{x}_2\,\tilde{x}_1^2\,p^{G}_{NF}(\tilde{x}_1,\tilde{x}_2)\nonumber\\
&=&\frac{1}{4}\left(\sigma+\tau^2\right)\label{varnfgaus},
\end{eqnarray}
and
\begin{eqnarray}
\langle \tilde{x}_1\tilde{x}_2\rangle_{G}=\langle \tilde{x}_2\tilde{x}_1\rangle_{G}&=&\int\!\!\!\int d\tilde{x}_1 d\tilde{x}_2\,\tilde{x}_1\tilde{x}_2\,p^{G}_{NF}(\tilde{x}_1,\tilde{x}_2)\nonumber\\
&=&\frac{1}{4}\left(\sigma-\tau^2\right)\label{covarnfgaus}.
\end{eqnarray}

From Eq. (\ref{varnf}), (\ref{covarnf}), (\ref{varnfgaus}) and (\ref{covarnfgaus}) we can observe that the condition
\begin{equation}
\tau^2=\frac{4A_2P^2}{A_1},
\end{equation}
must be satisfied for having $\Lambda_{NF}^G=\Lambda_{NF}^S$. The Gaussian distribution which satisfies this condition is denoted by $p^{\tilde{G}}_{NF}(\tilde{x}_1,\tilde{x}_2)$.

The differential entropy of this Gaussian distribution is given by
\begin{equation}
H[p^{\tilde{G}}_{NF}(\tilde{x}_1,\tilde{x}_2)]=\log_2\left(2\pi e\sqrt{\frac{\sigma A_2}{A_1}}P\right).
\end{equation}

The differential entropy for $p^S_{NF}(\tilde{x}_1,\tilde{x}_2)$ reads
\begin{eqnarray}
H[p^S_{NF}(\tilde{x}_1,\tilde{x}_2)]&=&-\int\!\!\!\int d\tilde{x}_1 d\tilde{x}_2\,p^S_{NF}(\tilde{x}_1,\tilde{x}_2)\log_2(p^S_{NF}(\tilde{x}_1,\tilde{x}_2))\nonumber\\
&=&\log_2(\sqrt{2\pi\sigma}A_1P)-\frac{1}{\sqrt{2\pi\sigma}A_1P}\int\!\!\!\int d\tilde{x}_1 d\tilde{x}_2\,\left\lbrace e^{-\frac{(\tilde{x}_1+\tilde{x}_2)^2}{2\sigma}}\right.\nonumber\\
& &\times\left.\sint^2\frac{(\tilde{x}_1-\tilde{x}_2)^2}{4P^2}\left[\log_2e^{-\frac{(\tilde{x}_1+\tilde{x}_2)^2}{2\sigma}}\right.\right.\nonumber\\
& &\left.\left.+\log_2\sint^2\frac{(\tilde{x}_1-\tilde{x}_2)^2}{4P^2}\right]\right\rbrace.
\end{eqnarray}

If $r=\frac{\tilde{x}_1+\tilde{x}_2}{\sqrt{2\sigma}}$ and $s=\frac{\tilde{x}_1-\tilde{x}_2}{2P}$, we find that
\begin{eqnarray}
H[p^S_{NF}(\tilde{x}_1,\tilde{x}_2)]&=&\log_2(\sqrt{2\pi\sigma}A_1P)\nonumber\\
& & +\frac{1}{2\ln 2}\left[1-\frac{2}{A_1}\int ds\,\sint^2 s^2\ln\sint^2 s^2\right].
\end{eqnarray}

The integral $$I=\int ds\sint^2 s^2\ln\sint^2 s^2,$$ must be calculated numerically. It gives that $I\approx-0.692$. Then, the differential entropy of the near-field joint distribution is
\begin{equation}
H[p^S_{NF}(\tilde{x}_1,\tilde{x}_2)]\approx\log_2(\sqrt{2\pi\sigma}A_1P)+1.434.
\end{equation}

Therefore, the negentropy of the near-field joint distribution is
\begin{eqnarray}
N[p^S_{NF}(\tilde{x}_1,\tilde{x}_2)]&=&H[p^{\tilde{G}}_{NF}(\tilde{x}_1,\tilde{x}_2)]-H[p^S_{NF}(\tilde{x}_1,\tilde{x}_2)]\nonumber\\
&\approx& \log_2\left(2\pi e\sqrt{\frac{\sigma A_2}{A_1}}P\right)-\log_2(\sqrt{2\pi\sigma}A_1P)-1.434\nonumber\\
&\approx& \log_2\left(\sqrt{2\pi} e\sqrt{\frac{A_2}{A_1^3}}\right)-1.434\nonumber\\
&\approx& 0.22.
\end{eqnarray}

Thus, the total non-Gaussianity of SPDC is
\begin{equation}
nG^T\approx 0.15+0.22=0.37.
\end{equation}

\section{Further properties of $nG^T$\label{SuplementoD}}
\newtheorem{theorem}{Property}
\begin{theorem}
$nG^T=0$ iff the spatial two-photon state of SPDC is a Gaussian state.
\end{theorem}
\begin{proof}
Since
\begin{equation}
nG^T=N_{NF}+N_{FF}=0,
\end{equation}
and that negentropy is always nonnegative, we have that $N[p^S_{NF}(\tilde{x}_1,\tilde{x}_2)]=N[p^S_{FF}(\tilde{q}_1,\tilde{q}_2)]=0$. Thus, for the joint probabilities we have that $H[p^{\tilde{G}}_{FF}(\tilde{q}_1,\tilde{q}_2)]=H[p^S_{FF}(\tilde{q}_1,\tilde{q}_2)]$ and $H[p^{\tilde{G}}_{NF}(\tilde{x}_1,\tilde{x}_2)]=H[p^S_{NF}(\tilde{x}_1,\tilde{x}_2)]$, such that
\begin{eqnarray}
\int d\tilde{q}_1 d\tilde{q}_2\,p^{\tilde{G}}_{FF}(\tilde{q}_1,\tilde{q}_2) \log_2 p^{\tilde{G}}_{FF}(\tilde{q}_1,\tilde{q}_2)=  \nonumber \\ \int d\tilde{q}_1 d\tilde{q}_2\, p^S_{FF}(\tilde{q}_1,\tilde{q}_2)\log_2 p^S_{FF}(\tilde{q}_1,\tilde{q}_2),
\end{eqnarray}
and
\begin{eqnarray}
\int d\tilde{x}_1 d\tilde{x}_2\, p^{\tilde{G}}_{NF}(\tilde{x}_1,\tilde{x}_2) \log_2 p^{\tilde{G}}_{NF}(\tilde{x}_1,\tilde{x}_2)= \nonumber \\ \int d\tilde{x}_1 d\tilde{x}_2\, p^S_{NF}(\tilde{x}_1,\tilde{x}_2) \log_2 p^S_{NF}(\tilde{x}_1,\tilde{x}_2).
\end{eqnarray}
Since the differential entropy $H$ is a continuous and monotonic function, it holds that
\begin{eqnarray}
\!p^{\tilde{G}}_{FF}(\tilde{q}_1,\tilde{q}_2)\!\log_2 p^{\tilde{G}}_{FF}(\tilde{q}_1,\tilde{q}_2)\!= \nonumber \\ \!p^S_{FF}(\tilde{q}_1,\tilde{q}_2)\! \log_2 p^S_{FF}(\tilde{q}_1,\tilde{q}_2),
\end{eqnarray}
and
\begin{eqnarray}
p^{\tilde{G}}_{NF}(\tilde{x}_1,\tilde{x}_2)\!\log_2 p^{\tilde{G}}_{NF}(\tilde{x}_1,\tilde{x}_2)\!= \nonumber \\ \!p^S_{NF}(\tilde{x}_1,\tilde{x}_2)\!\log_2 p^S_{NF}(\tilde{x}_1,\tilde{x}_2).
\end{eqnarray}
Then, necessarily we have that
\begin{equation}
p^{\tilde{G}}_{FF}(\tilde{q}_1,\tilde{q}_2)=p^S_{FF}(\tilde{q}_1,\tilde{q}_2),
\end{equation} and
\begin{equation}
p^{\tilde{G}}_{NF}(\tilde{x}_1,\tilde{x}_2)=p^S_{NF}(\tilde{x}_1,\tilde{x}_2).
\end{equation}
Thus, we conclude that the two-photon state is a Gaussian state.

In the opposite direction we consider that the two-photon state is a Gaussian state. Then, the joint probability density functions of transverse position and momentum are Gaussian functions. So, from the definition of negentropy we must have that
\begin{equation}
N[p^S_{NF}(\tilde{x}_1,\tilde{x}_2)]=N[p^S_{FF}(\tilde{q}_1,\tilde{q}_2)]=0.
\end{equation}

So, the total non-Gaussianity of the spatial two-photon state becomes
\begin{equation}
nG^T=0.
\end{equation}
\end{proof}

\begin{theorem}
If the spatial two-photon state is a product state then $nG^T$ is additive, i.e., $nG^T=2\times nG^M$.
\end{theorem}
\begin{proof}
For a set of two random and independent variables, the joint probability density functions are given by the product of the probability density functions associated to each variable, i.e., $p(\xi_1,\xi_2)=p(\xi_1)p(\xi_2)$. Such type of probability density function has the covariance matrix $\Lambda_{ij}$ defined by the elements
\begin{eqnarray}
\Lambda_{11}&=&\esp{\xi_1^2}-\esp{\xi_1}^2\nonumber\\
&=&\int d\xi_1\,\xi_1^2p(\xi_1)-\left(\int d\xi_1\,\xi_1 p(\xi_1)\right)^2,\\
\Lambda_{22}&=&\esp{\xi_2^2}-\esp{\xi_2}^2\nonumber\\
&=&\int d\xi_2\,\xi_2^2p(\xi_2)-\left(\int d\xi_2\,\xi_2 p(\xi_2)\right)^2,\\
\Lambda_{12}&=&\esp{\xi_1\xi_2}-\esp{\xi_1}\esp{\xi_2}\nonumber\\
&=&\int d\xi_1\,\xi_1 p(\xi_1)\int d\xi_2\,\xi_2p(\xi_2)-\int d\xi_1\,\xi_1p(\xi_1)\int d\xi_2\,\xi_2p(\xi_2)\nonumber\\
&=&0,\\
\Lambda_{21}&=&\Lambda_{12}=0,
\end{eqnarray}
and a expected value vector $\boldsymbol{\mu}=\{\esp{\xi_1},\esp{\xi_2}\}$.

The Gaussian distribution $p^{\tilde G}(\xi_1,\xi_2)$ with the same covariance matrix and expected value vector is given by
\begin{eqnarray}
p^{\tilde G}(\xi_1,\xi_2)&=&\frac{1}{2\pi\sqrt{\det\Lambda}}\exp\left[-\frac{1}{2}(\boldsymbol{\xi}-\boldsymbol{\mu})\Lambda^{-1}(\boldsymbol{\xi}-\boldsymbol{\mu})^T\right]\nonumber\\
&=&\frac{e^{-\frac{(\xi_1-\esp{\xi_1})^2}{2\Lambda_{11}}}}{\sqrt{2\pi\Lambda_{11}}}\times\frac{e^{-\frac{(\xi_2-\esp{\xi_2})^2}{2\Lambda_{22}}}}{\sqrt{2\pi\Lambda_{22}}}\nonumber\\
&=&p^{\tilde G}(\xi_1)p^{\tilde G}(\xi_2),
\end{eqnarray}
where $\boldsymbol{\xi}=\{\xi_1,\xi_2\}$.

Note that these marginal Gaussian distributions have the same variance and expected value that $p(\xi_1)$ and $p(\xi_2)$. Since the random variables are independent, it holds that
\begin{equation}
H[p(\xi_1,\xi_2)]=H[p(\xi_1)]+H[p(\xi_2)],
\end{equation}
and consequently, we have that the negentropy of $p(\xi_1,\xi_2)$ can be written as
\begin{eqnarray}
N[p(\xi_1,\xi_2)]&=&H[p^{\tilde G}(\xi_1,\xi_2)]-H[p(\xi_1,\xi_2)]\nonumber\\
&=&H[p^{\tilde G}(\xi_1)]\!-\!H[p(\xi_1)]\!+\!H[p^{\tilde G}(\xi_2)]\!-\!H[p(\xi_2)]\nonumber\\
&=&N[p(\xi_1)]+N[p(\xi_2)].
\end{eqnarray}

If the two-mode spatial state of SPDC is a product state, the near- and far-field probability density functions are written as $p^S_{FF}(\tilde{q}_1,\tilde{q}_2)=p^S_{FF}(\tilde{q}_1)p^S_{FF}(\tilde{q}_2)$ and $p^S_{NF}(\tilde{x}_1,\tilde{x}_2)=p^S_{NF}(\tilde{x}_1)p^S_{NF}(\tilde{x}_2)$. Therefore, the properties mentioned above hold. Then, the total non-Gaussianity of the spatial two-photon state of SPDC will be given by
\begin{eqnarray}
nG^T&=&N[p^S_{FF}(\tilde{q}_1,\tilde{q}_2)]+N[p^S_{NF}(\tilde{x}_1,\tilde{x}_2)]\nonumber\\
&=&N[p^S_{FF}(\tilde{q}_1)]+N[p^S_{FF}(\tilde{q}_2)]+N[p^S_{NF}(\tilde{x}_1)]+N[p^S_{NF}(\tilde{x}_2)]\nonumber\\
&=&nG^{M_1}+nG^{M_2},
\end{eqnarray}
where we have used that $nG^{M_1}\equiv N[p^S_{FF}(\tilde{q_1})]+N[p^S_{NF}(\tilde{x}_1)]$ and $nG^{M_2}\equiv N[p^S_{FF}(\tilde{q_2})]+N[p^S_{NF}(\tilde{x}_2)]$.

Due to the symmetry present in the two-photon wave function [see Eq. (3) and Eq. (4) of the main paper], we have that $nG^{M_1}=nG^{M_2}=nG^M$. Then
\begin{equation}
nG^T=2\times nG^M,
\end{equation}
which proves our statement.
\end{proof}

\begin{theorem}
The non-Gaussianity of the spatial state of SPDC decreases under partial trace, such that $nG^T > nG^M$.
\end{theorem}
\begin{proof}
To demonstrate this we first study the negentropies of the far- and near-field marginal distributions, while considering two limiting cases: (i) $P\ll1$ and (ii) $P$ very large. The joint distributions of the far- and near-field planes are given by Eq.~(\ref{psff}) and Eq.~(\ref{psnf}), respectively. Thus, the marginal distributions are given by

\begin{equation}
p^S_{FF}(\tilde{q}_2)=\int d\tilde{q}_1\,p^S_{FF}(\tilde{q}_1,\tilde{q}_2),
\end{equation} and

\begin{equation}
p^S_{NF}(\tilde{x}_2)=\int d\tilde{x}_1\,p^S_{NF}(\tilde{x}_1,\tilde{x}_2).
\end{equation}

Doing the following change of variables: $u=\frac{P}{2}(\tilde{q}_1-\tilde{q}_2)$ and $v=\frac{\tilde{x}_1-\tilde{x}_2}{2P}$, we obtain that

\begin{equation}\label{margff}
p^S_{FF}(\tilde{q}_2)\propto\int du\,e^{-\frac{2}{P^2}(u+P\tilde{q}_2)^2}\sinc^2u^2,
\end{equation} and

\begin{equation}\label{margnf}
p^S_{NF}(\tilde{x}_2)\propto\int dv\,e^{-2P^2\left(v+\frac{\tilde{x}_2}{P}\right)^2}\sint^2v^2.
\end{equation}

Considering the limiting case (i) we may write the marginal distributions as

\begin{equation}
p^S_{FF}(\tilde{q}_2)\propto\int du\,\delta(u+P\tilde{q}_2)\sinc^2u^2\propto\sinc^2(P^2\tilde{q}_2),
\end{equation} and

\begin{equation}
p^S_{NF}(\tilde{x}_2)\propto e^{-2\tilde{x}_2^2}.
\end{equation} In such case, it is clear that the negentropy of $p^S_{NF}(\tilde{x}_2)$ is null [$N[p^S_{NF}(\tilde{x}_2)]=0$]. The negentropy of the far-field marginal distribution will be the same of the function $\sinc^2(P^2\tilde{q}_2)$. It gives that

\begin{equation}
N[p^S_{FF}(\tilde{q}_2)]\approx\log_2\left(\sqrt{\frac{27e}{32}}\right)-0.4447=0.154.
\end{equation} Then, for values of $P\ll1$, $nG^M \approx 0.154$ such that $nG^T > nG^M$.

In the other limiting case (ii), we have that

\begin{equation}
p^S_{FF}(\tilde{q}_2)\propto e^{-2\tilde{q}_2^2},
\end{equation} and

\begin{equation}
p^S_{NF}(\tilde{x}_2)\propto\int dv\,\delta\left(v+\frac{\tilde{x}_2}{P}\right)\sint^2v^2\propto\sint^2\frac{\tilde{x}_2^2}{P^2}.
\end{equation} Now, the negentropy of $p^S_{FF}(\tilde{q}_2)$ is null [$N[p^S_{FF}(\tilde{q}_2)]=0$]. The negentropy of the near-field marginal distribution will be the same of the function $\sint^2\frac{\tilde{x}_2^2}{P^2}$. It gives that

\begin{equation}
N[p^S_{NF}(\tilde{x}_2)]\approx\log_2\left(\sqrt{\frac{2\pi e A_2}{A_1^3}}\right)-0.7122 = 0.224
\end{equation}  Thus, for larger values of $P$, $nG^M \approx 0.224$ which is also smaller than $nG^T$.

Intermediate cases, for values of $0.002<P<3$, have been calculated numerically. The values of the negentropies of the far- and near-field marginal distributions are shown in Figure 3(b) of the main paper. In Figure 3(d), $nG^M$ is plotted in terms of $P$. As one can see, in such intermediate cases, one always have that $nG^T > nG^M$. For $P=3$ the value of $nG^M\approx0.175$. For values of $P>9$, $nG^M$ already tends to its maximal value which is $nG^M_{max}=0.224$ (See Figure \ref{Fig2sup} below). This proves that the non-Gaussianity of the spatial state of SPDC decreases under partial trace.
\end{proof}

\begin{figure}[h]
\centering
\includegraphics[angle=-90,width=0.6\textwidth]{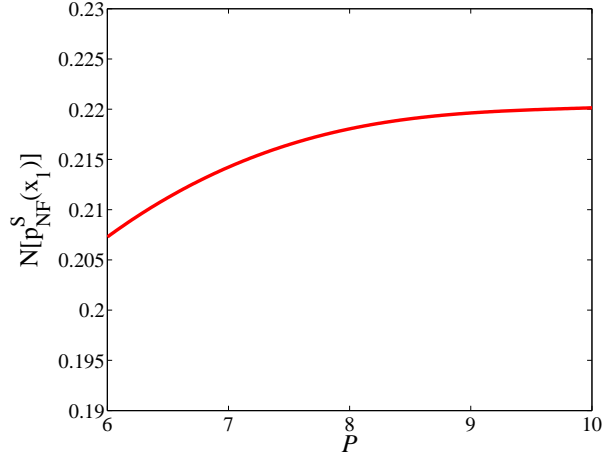}
\caption{Here we show the non-Gaussianity of the marginal distributions for larger values of $P$.}\label{Fig2sup}
\end{figure}

\section{Calculating the non-Gaussianity of the state of the spatially correlated down-converted photons using the QRE \label{SuplementoE}}

For comparison purposes, we calculated the value of the non-Gaussianity measure based on the QRE, $\delta_B$ \cite{Banaszek,Paris}. This was done by assuming the case of perfect phase-matching to simplify the calculation of the mean values in the covariance matrix. Then, we considered that no correlations are present between the position of one photon and the momentum of the other photon. Therefore, the cross joint probabilities $p(x_i,q_j)$ (where $i,j=1,2$ and $i\neq j$) are described by a product of the marginal position and momentum distributions $p(x_i)$ and $p(q_j)$. In such case the covariance matrix can be written, considering the order $\mathbf{V}(\tilde{x}_1,\tilde{q}_1,\tilde{x}_2,\tilde{q}_2)$, as

\begin{equation}
\mathbf{V} = \begin{pmatrix}
\frac{1}{4}\left(1 + \frac{4A_2 P^2}{A_1}\right) & 0 & \frac{1}{4}\left(1 - \frac{4A_2 P^2}{A_1}\right) & 0\\
0 & \frac{1}{4}\left(1 + \frac{3}{P^2}\right) & 0 &\frac{1}{4}\left(1 - \frac{3}{P^2}\right)\\
\frac{1}{4}\left(1 - \frac{4A_2 P^2}{A_1}\right) & 0 & \frac{1}{4}\left(1 + \frac{4A_2 P^2}{A_1}\right) & 0\\
0 & \frac{1}{4}\left(1 - \frac{3}{P^2}\right) & 0 &\frac{1}{4}\left(1 + \frac{3}{P^2}\right)
\end{pmatrix}.
\end{equation}

The Gaussian state associated to this covariance matrix has a purity of $\mu=0.44$ \cite{Paris}, and its Von-Neumann entropy reads $S=1.08$. Since the state of Eq~(\ref{EstMon}) is pure, we have that $\delta_B=1.08$, which does not depend on the value of $P$.

\section{$\mathbf{nG^T}$ at the limit when $\mathbf{P\ll 1}$\label{SuplementoF}}
Here we discuss the behavior of $nG^T$ when considering the limit of $P\ll1$, and that one of the down-converted photon is detected in transverse points around the origin. First we analyze the negentropy of far-field joint distribution:
\begin{equation}
N[p^S_{FF}(\tilde{q}_1,\tilde{q}_2)]=H[p^{\tilde G}_{FF}(\tilde{q}_1,\tilde{q}_2)]-H[p^S_{FF}(\tilde{q}_1,\tilde{q}_2)].
\end{equation}
In the limits considered, the differential entropies can be written as $H[p(\xi_1,\xi_2)]=H[p(\xi_1|\xi_2)]+H[p(\xi_2)]$, such that the negentropy of the far-field joint distribution may also be rewritten as
\begin{eqnarray}
N[p^S_{FF}(\tilde{q}_1,\tilde{q}_2)]&=&H[p^{\tilde G}_{FF}(\tilde{q}_1|\tilde{q}_2)]+H[p^{\tilde G}_{FF}(\tilde{q}_2)]\nonumber\\
& &-H[p^S_{FF}(\tilde{q}_1|\tilde{q}_2)]-H[p^S_{FF}(\tilde{q}_2)]\nonumber\\
&=&H[p^{\tilde G}_{FF}(\tilde{q}_1|\tilde{q}_2)]-H[p^S_{FF}(\tilde{q}_1|\tilde{q}_2)]\nonumber\\
& &+N[p^S_{FF}(\tilde{q}_2)].\label{nffc}
\end{eqnarray}
Note that $p^{\tilde G}_{FF}(\tilde{q}_1|\tilde{q}_2)$ does not have the same expected value and variance of the conditional probability density function $p^S_{FF}(\tilde{q}_1|\tilde{q}_2)$. Thus, in general, we can not define $H[p^{\tilde G}(\tilde{q}_1|\tilde{q}_2)]-H[p^S_{FF}(\tilde{q}_1|\tilde{q}_2)]$ as the negentropy of the conditional far-field distribution.

Considering the joint far-field distribution (\ref{psff}), the conditional distribution when $\tilde{q}_2=0$ will be given by
\begin{equation}
p^S_{FF}(\tilde{q}_1|\tilde{q}_2=0)\propto e^{-\frac{\tilde{q}_1^2}{2}}\sinc^2\frac{P^2\tilde{q}_1^2}{4}.
\end{equation}
When $P\ll 1$, the $\sinc$ function is much larger than the Gaussian function, such that we can approximate this distribution as
\begin{equation}\label{ffappr}
p^S_{FF}(\tilde{q}_1|\tilde{q}_2=0)\approx\frac{1}{\sqrt{2\pi}}e^{-\frac{\tilde{q}_1^2}{2}}.
\end{equation} Therefore, the differential entropy is
\begin{equation}\label{hpffcon}
H[p^S_{FF}(\tilde{q}_1|\tilde{q}_2=0)]\approx\frac{1}{2}\log_2(2\pi e).
\end{equation} Note that Eq.~(\ref{ffappr}) is a normal distribution with the expected value equal to zero and the variance equal to 1.

Now let's consider the Gaussian distribution $p^{\tilde G}_{FF}(\tilde{q}_1,\tilde{q}_2)$ of Eq.~(\ref{pgff}). It gives the conditional Gaussian distribution $p^{\tilde G}_{FF}(\tilde{q}_1|\tilde{q}_2)$ at $\tilde{q_2}=0$ equal to
\begin{equation}
p^{\tilde G}_{FF}(\tilde{q}_1|\tilde{q}_2=0)=\frac{\sqrt{3+P^2}}{\sqrt{6\pi}}e^{-\tilde{q}_1^2\frac{3+P^2}{6}}.
\end{equation}

In the case when $P\ll 1$, this Gaussian distribution has also the expected value equal to zero and the variance equal to 1. Thus, it corresponds to the Gaussian distribution with the same expected value and variance of $p^S_{FF}(\tilde{q}_1|\tilde{q}_2=0)$. Hence, \emph{in this limit}, we can define the negentropy of the far-field conditional distribution in terms of the $p^{\tilde G}_{FF}(\tilde{q}_1|\tilde{q}_2=0)$, such that
\begin{equation}
N[p^S_{FF}(\tilde{q}_1,\tilde{q}_2)]=N[p^S_{FF}(\tilde{q}_1|\tilde{q}_2)]+N[p^S_{FF}(\tilde{q}_2)].
\end{equation}

Now we consider the negentropy of the near-field joint distribution (\ref{psnf}). If $P\ll 1$, then the near-field conditional distribution is given by
\begin{equation}
p^S_{NF}(\tilde{x}_1|\tilde{x}_2=0)\approx\frac{1}{2PA_1}\sint^2\frac{\tilde{x}_1^2}{4P^2},
\end{equation} where we assumed $\sigma=1$ for simplicity.

The variance of this function is
\begin{equation}
(\Delta \tilde{x}_1|_{\tilde{x}_2=0})^2=\int d\tilde{x}_1 \,\tilde{x}_1^2p^S_{NF}(\tilde{x}_1|\tilde{x}_2=0)=\frac{4A_2 P^2}{A_1},
\end{equation} where $A_1$ and $A_2$ are the constants defined above. Now we consider the near-field Gaussian joint distribution $p^{\tilde G}_{NF}(\tilde{x}_1,\tilde{x}_2)$. The near-field conditional distribution associated with this Gaussian approximation, when $P\ll1$, is given by
\begin{equation}
p^{\tilde G}_{NF}(\tilde{x}_1|\tilde{x}_2=0)=\frac{\sqrt{A_1}}{2\sqrt{2\pi A_2}P}e^{-\frac{A_1}{8A_2P^2}\tilde{x}_1^2}.
\end{equation}

As one can note, $p^{\tilde G}_{NF}(\tilde{x}_1|\tilde{x}_2=0)$ has the same expected value and variance of $p^S_{NF}(\tilde{x}_1|\tilde{x}_2=0)$. Therefore, \emph{in this limit} we can write the negentropy of the near-field joint distribution as
\begin{equation}
N[p^S_{NF}(\tilde{x}_1,\tilde{x}_2)]=N[p^S_{NF}(\tilde{x}_1|\tilde{x}_2)]+N[p^S_{NF}(\tilde{x}_2)],
\end{equation} where $N[p^S_{NF}(\tilde{x}_1|\tilde{x}_2)]$ is defined in terms of the differential entropy of $p^{\tilde G}_{NF}(\tilde{x}_1|\tilde{x}_2)$.

Thus, when we consider $P\ll 1$ we can write the total non-Gaussianity of the spatial two-photon state of SPDC as
\begin{equation}
nG^T\approx nG^C+nG^M \textrm{  (iff $P\ll 1$)},
\end{equation} where $nG^C=N[p^S_{NF}(\tilde{x}_1|\tilde{x}_2)]+N[p^S_{FF}(\tilde{q}_1|\tilde{q}_2)]$, and $nG^M=N[p^S_{FF}(\tilde{q_j})]+N[p^S_{NF}(\tilde{x}_j)]$, with $j=1,2$.


\section*{Acknowledgments}
We thank M. Martinelli, J. L. Romero, C. Saavedra and L. Neves for discussions related with earlier drafts of this paper. We acknowledge the support of Grants FONDECYT 1120067, Milenio~P10-030-F and PFB~08024. C. H. M. acknowledges CNPq and CAPES. E. S. G. acknowledges the financial support of CONICYT.
\end{document}